\documentclass{amsproc}
\usepackage[colorlinks=true,linkcolor=blue,citecolor=blue]{hyperref}
\usepackage[foot]{amsaddr}

\newcommand{\Datalog}{\text{\sf Datalog}}

\allowdisplaybreaks
\usepackage{fullpage}

\usepackage{bm} 
\usepackage{url}
\usepackage{comment}

\usepackage[colorinlistoftodos]{todonotes}

\usepackage{enumitem}
\usepackage{subcaption}

\usepackage{microtype}

\newcommand{\setof}[2]{\{{#1}\mid{#2}\}}        

\newcommand{\calW}{\mathcal W}

\newcommand{\calM}{\mathcal M}

\newcommand{\calL}{\mathcal L}

\newcommand{\calT}{\mathcal T}

\newtheorem{thm}{Theorem}[section]
\newtheorem{lmm}[thm]{Lemma}

\newtheorem{defn}[thm]{Definition}

\newtheorem{exnum}{Example}
\newtheorem{ex}[exnum]{Example}

\newcommand{\N}{\mathbb N} 
\newcommand{\R}{\mathbb R} 

\newcommand{\cd}{\text{ :- }}

\newcommand{\Name}{\text{\sf Datalog}^\circ}
\newcommand{\trop}{\text{\sf Trop}}

\newcommand{\adom}{\textsf{ADom}}

\begin{document}
\title{Polynomial Time Convergence of the Iterative Evaluation of Datalogo Programs}

\author{Sungjin Im}
\address{UC Merced}
\author{Ben Moseley}
\address{CMU}

\author{Hung Q. Ngo}
\address{RelationalAI, Inc}

\author{Kirk Pruhs}
\address{University of Pittsburg}

\begin{abstract}
 $\Name$ is
an extension of $\Datalog$ that allows for aggregation and recursion
over an arbitrary {\em commutative semiring}.
Like $\Datalog$, $\Name$ programs can be evaluated via the natural iterative algorithm until a fixed point is reached.
However unlike $\Datalog$, the natural iterative
evaluation of
some $\Name$ programs over some semirings may not  converge.
It is known that the commutative semirings for which the iterative evaluation of  $\Name$ programs is guaranteed to converge
are exactly those semirings that are stable~\cite{Khamis0PSW22}.
Previously, the best known upper bound on the number of iterations
until convergence  over $p$-stable semirings is $\sum_{i=1}^n (p+2)^i = \Theta(p^n)$ steps, where $n$ is  (essentially)
the output size.
We establish that, in fact, the
natural iterative evaluation of a $\Name$ program over a $p$-stable semiring converges within a polynomial
number of iterations. In particular our upper bound is
 $O( \sigma p n^2( n^2 \lg \lambda + \lg \sigma))$ where
 $\sigma$ is the number of elements in the semiring present
 in either the input databases or the $\Name$ program, and $\lambda$ is the
 maximum number of terms in any product in the $\Name$ program.
\end{abstract}


\maketitle
\section{Introduction}

Motivated by the need in modern data analytics to express recursive computations with
aggregates, Khamis et al.~\cite{Khamis0PSW22} introduced $\Name$, which is
an extension of $\Datalog$ that allows for aggregation and recursion
over an arbitrary {\em commutative semiring}.\footnote{{The results in~\cite{Khamis0PSW22}
are on {\em Partially Ordered Pre-Semirings} (POPS). However, the key convergence properties  are reflected in the {\em core semiring} of the POPS. Thus, it is sufficient to  restrict our attention to semirings for the purpose of this paper.}}
Like $\Datalog$, $\Name$ programs can be evaluated via the natural iterative
algorithm until a fixed point is reached. This is sometimes called the
``na\"ive evaluation'' algorithm.
Furthermore, $\Name$ is attractive for practical applications
because it also allows for a
generalization of semi-na\"ive evaluation to work, under some assumptions
about the semiring~\cite{Khamis0PSW22}. While semi-na\"ive evaluation
makes each iteration faster to compute, the total number of iterations
is the same as that of the na\"ive evaluation algorithm.
Thus, bounding the number of iterations of the na\"ive evaluation is an important
question in practice.

In $\Datalog$, it is easy to see that
the number of iterations until a fixed point is reached is at most the output size. (Every iteration before convergence must derive at least one
new fact, due to monotonicity.)
In contrast, the na\"ive evaluation of $\Name$ programs over some commutative semirings may
not converge. (A simple example is the sum-product semiring over the reals.)

It is known that the commutative semirings for which the iterative evaluation of  $\Name$ programs is guaranteed to converge
are exactly those semirings that are stable~\cite{Khamis0PSW22}.
A semiring is {\em $p$-stable}~\cite{semiring_book} if the  number of iterations required for any one-variable recursive linear $\Name$ program to reach
a fixed point is at most $p$, and a semiring is stable if there exists a $p$ for which it is $p$-stable.
Previously, the best known upper bound on the number of iterations
until convergence  is $\sum_{i=1}^n (p+2)^i = \Theta((p+2)^n)$ steps, where $n$ is  (essentially)
the output size, and $p$ is the stability index of the underlying semiring.
In contrast there are no known lower bounds that show that  iterative evaluation requires an exponential (in the parameter $n$) number of steps to reach convergence.

There are special cases where polynomial convergence rate is known.
The first case is when the semiring is $0$-stable, as in the standard Boolean semiring,
where it is known that na\"ive evaluation converges in $O(n)$ steps~\cite{Khamis0PSW22}.
The second case is when the input program is {\em linear}, meaning that in every rule the
product only contains at most one IDB relational symbol.
In \cite{im2023convergence} it is shown that if the semiring is $p$-stable with $L$ elements in the semiring domain,
then the iterative evaluation of all linear $\Name$ programs converge after $O(\min(p n^3, p n \lg L))$ steps.

\smallskip \textbf{Our Contributions.}
The open problem we address in this paper is whether the iterative evaluation
of $\Name$ programs over $p$-stable semirings
might indeed require an exponential number of steps to converge.  Another way to
frame our motivating research question is whether or not polynomial convergence is a special
property of {\em linear} $\Name$ programs that is not shared by general $\Name$ programs.
Our main finding is stated in Theorem \ref{thm:main}.

\begin{thm}
\label{thm:main}
Let $\bm S$ be a $p$-stable commutative semiring.
Let $P$ be a $\Name$ program where the maximum number of multiplicands
in any product is at most $\lambda$. Let $D$ be the input database instance.
Let $\sigma$  be number of the semiring elements referenced in $P$ or $D$.
Let $n$ denote the total number of ground atoms in an IDB
that at some point in the iterative evaluation of $P$ over $\bm S$ on input $D$ have
a nonzero associated semiring value.
Then the iterative evaluation of $P$ over $\bm S$ on input $D$
    converges within
    $$\lceil  p n(n+3) \cdot (\sigma (n(n+3) / 2) \lg (\lambda+1) + 4\sigma \lg \sigma + 1) \rceil$$
        steps.
\end{thm}

{
As $\lambda$ is only a property of the $\Name$ program, and $p$ is only a property
of the semiring, they do not scale with the data size.
Thinking of them as constants in data complexity,
the bound in the Theorem \ref{thm:main} is reduced to
$O(n^4 \sigma \lg \sigma)$.
Note that $\sigma$ is bounded by the input size plus the query size and so
it is linear in  the input size under combined
complexity.\footnote{Note that in \cite{Khamis0PSW22} the parameter $\sigma$ denoted the number of references to semiring elements, not the number of semiring elements referenced.} The maximum number of ground IDB atoms is $\tilde O(|D|^k)$ where $k$ is the
maximum arity of IDB atoms, where $\tilde O$ hides query-dependent factors.
Thus, overall, in data complexity Theorem~\ref{thm:main} gives a polynomial bound
on the convergence rate; furthermore, its dependency on the output size makes the
bound more flexible.

\smallskip \textbf{Related Works.}
There is a large body   of research on
fixed points of multi-variate polynomial functions over semirings, which were
studied by many communities since the 1960s.
(See e.g.~\cite{MR1470001,Khamis0PSW22,DBLP:journals/jacm/EsparzaKL10,MR1059930}.)
In some special cases, such as closed semirings or $\omega$-continuous semirings,
there are non-iterative methods to find the
fixpoint~\cite{DBLP:journals/tcs/Lehmann77,DBLP:journals/jacm/EsparzaKL10}.
In general, the iterative algorithm is still the most general.

The two   papers in the literature that we directly build on are~\cite{Khamis0PSW22} and~\cite{im2023convergence}.  In \cite{Khamis0PSW22} it is shown that  the na\"ive evalution of $\Name$ programs converges in
$O((p+2)^n)$ steps; this bound is obtained by showing how to bound the convergence time
for a high dimensional function in terms of the convergence time for a 1-dimensional function.

The paper \cite{im2023convergence} considers {\em linear}
$\Name$ programs, where the grounded ICO is a linear function
$\bm f : S^n \rightarrow S^n$ that can expressed as
$\bm f(x) = A \otimes x \oplus b$ where $A$ is an $n$ by $n$ matrix with entries from the semiring.
Then $f_i^{(q)}(0)$ becomes $\oplus_{W \in \calW^i_q} Z(W)$,
where $\calW^i_q$ is the collection of all walks starting from ground atom $i$ with at most $q$ hops
in the natural complete digraph underlying $A$, and $Z(W)$ is the
product of $a_{uv}$ for every edge $uv \in W$.
The crux of the $O(p n^3)$ upper bound
analysis  in \cite{im2023convergence} was then that any walk $W$ longer than $O(p n^3)$ must contain a
cycle $C$ where all the edges are traversed many times.
The fact that adding $Z(W)$ didn't change the sum followed
from the stability of the semiring element that is the product of the semiring values on the edges in $C$.
Our analysis for the nonlinear case is more involved than the analysis
for the linear cases because
finding the collection of semiring values that will serve the role of the cycle $C$ in the linear case is more involved.
It is interesting to note that, however, the gap between the two cases is
$O(n\sigma \lg \sigma)$.
}

Our main theorem is proved via a strengthening of  Parikh's Theorem~\cite{MR209093}, which we believe is novel\footnote{However, as there are a daunting number of different statements, proofs and extensions  of Parikh's theorem  in the literature~\cite{friendlyParikh}, it is hard to be totally confident about the
novelty of our extension.} and may be of independent
interests.

\smallskip \textbf{Paper Organization.} Section~\ref{sec:background} covers background knowledge required
to understand this result.
Section~\ref{sec:overview} gives a brief technical overview of the
proof of Theorem \ref{thm:main}.
Sections~\ref{sect:Parikh}, ~\ref{sect:smallsupport}, and~\ref{sect:mainproof} give the proof details.
Finally Section~\ref{sec:conclusions} concludes the paper.

\section{Background}
\label{sec:background}

\subsection{Semirings}
A {\em semiring} is a tuple $\bm S = (S, \oplus, \otimes, 0, 1)$ where
$\oplus$ and $\otimes$ are
binary operators on $S$,
$(S, \oplus, 0)$ is a commutative monoid
(meaning $\oplus$ is commutative and associative, and $0$ is the identity for $\oplus$),
$(S, \otimes, 1)$ is a monoid (meaning $\otimes$ is associative, and $0$ is the identity for $\oplus$),
$a\otimes 0 = 0 \otimes a = 0$ for every $a \in S$, and
$\otimes$ distributes over $\oplus$.
$\bm S$ is said to be {\em commutative} if $\otimes$ is commutative.
Define
$$u^{(p)} := 1 \oplus u \oplus u^2 \oplus \cdots \oplus u^{p},$$
where $u^{i} := u \otimes u \otimes \cdots \otimes u$ ($i$ times).
An element $u\in S$ is {\em $p$-stable} if $u^{(p)}=u^{(p+1)}$, and a
semiring $\bm S$ is {\em $p$-stable} if every element $u \in S$ is $p$-stable.

A function $f : S^n \to S^n$ is $p$-stable if $f^{(p+1)}(\bm 0) = f^{(p)}(\bm 0)$, where
$\bm 0$ is the all zero vector, and
$f^{(k)}$ is the $k$-fold composition of $f$ with itself.  The {\em stability index} of
$f$ is the smallest $p$ such that $f$ is $p$-stable.
See \cite{semiring_book} for more background on semirings and stability.

\subsection{Datalog}
A (traditional) $\Datalog$~\cite{DBLP:books/aw/AbiteboulHV95}
program  $P$ consists of a set of rules of the form:
\begin{align}
  R_0(\bm X_0) &\cd R_1(\bm X_1) \wedge \cdots \wedge R_m(\bm X_m) \label{eq:datalog}
\end{align}
where $R_0, \ldots, R_m$ are predicate names (not necessarily distinct) and each $\bm X_i$
is a tuple of variables and/or constants.   The atom $R_0(\bm X_0)$ is called the head, and
the conjunction $R_1(\bm X_1) \wedge \cdots \wedge R_m(\bm X_m)$ is called the body.
Multiple rules with the same head are interpreted as a disjunction.
A predicate that occurs in the head of some rule in $P$ is called an
{\em intensional database predicate} (IDB), otherwise it is called an
{\em extensional database predicate} (EDB).
The EDBs form the input,  and the  IDBs represent the output computed by
the $\Datalog$ program.
The finite set of all constants occurring in an EDB is called the
{\em active domain}, and denoted $\adom$.
An atom $R(\bm X)$ is called a {\em ground atom} if all its arguments are constants.
There is an implicit existential quantifier over the body for all variables
that appear in the body, but not in the head, where the domain of the existential quantifier is $\adom$. Thus, a Datalog program can also be viewed as a collection
of unions of conjunctive queries (UCQs), one UCQ for each IDB.

\begin{ex} \label{ex:TC}
A classic example of a $\Datalog$ program is the  transitive closure program
\begin{align*}
  T(X,Y) &\cd E(X,Y) \\
  T(X,Y) &\cd T(X,Z) \wedge E(Z,Y)
\end{align*}
Here $E$ is an EDB predicate, representing the edge relation of a directed graph,
$T$ is an IDB predicate, and $\adom$ is the vertex set.
Written as a UCQ, where the quantifications are explicit, this program is:
\begin{align}
  T(X,Y) & \cd E(X,Y) \vee \exists_Z \ (T(X,Z) \wedge E(Z,Y)) \label{eq:datalog:intro}
\end{align}

The UCQ format is the right format to work with when extending $\Datalog$ programs
to general semirings.
\end{ex}

A $\Datalog$ program $P$ can be thought of as a function,
called the {\em immediate consequence operator (ICO)},
mapping a subset of ground IDB atoms to a subset of ground IDB atoms.
(The ground EDB atoms are inputs and thus remain constants.)
In particular, the ICO adds a ground (IDB) atom $R(\bm x)$ to the output
if it can be logically inferred by the input ground atoms via the rules of $P$.
The iterative evaluation of a $\Datalog$ program works in rounds/steps,
where on each round the ICO is applied to the current state, starting from the empty
state.

\subsection{Datalogo}
Like $\Datalog$ programs, a $\Name$
program consists of a set of rules, where the UCQs are
replaced by {\em sum-sum-product queries} over a commutative semiring
$\bm S = (S, \oplus, \otimes, 0, 1)$, {where $\vee$ is replaced with
$\oplus$ and $\wedge$ with $\otimes$.
Specifically, in a $\Name program$ each rule has the form:}
\begin{align} R_0(X_0) &\cd \bigoplus R_1(\bm X_1) \otimes
    \cdots \otimes R_m(\bm X_m) \label{eq:t:monomial}
  \end{align}
where   sum is over the active $\adom$ of the variables not in $X_0$.
{Multiple rules with the same head are combined using the
$\oplus$ operation, which is the analog of combining rules using $\vee$
in $\Datalog$.}

Furthermore, each ground EDB or IDB atom is associated with an element of the semiring $\bm S$, and the non-zero elements associated with ground EDB atoms
are specified in the input.
A fixed point solution to the $\Name$ program associates a semiring
element to ground IDB atoms.
Just like in $\Datalog$, we do not have to explicitly represent the
zero-assigned ground IDB atoms: every ground atom not
in the output are implicitly mapped to $0$.

\begin{ex} \label{ex:APSP1}
The $\Name$-version of the $\Datalog$ program given in  line~\eqref{eq:datalog:intro} with
$\vee$ replaced by $\oplus$, $\wedge$ replaced by $\otimes$ and $\exists_Z$ replaced by $\bigoplus_Z$ is
\begin{align}
    T(X,Y) &\cd E(X,Y) \oplus \bigoplus_Z T(X,Z) \otimes E(Z,Y), \label{eqn:linear:tc}
\end{align}
Here $E$ is an EDB predicate, and $T$ is an IDB predicate,
$\oplus$ is the semiring addition operation, $\otimes$ is the semiring multiplication
operation
and $\bigoplus_Z$ is aggregation, that is an iterative application of $\oplus$
over $\adom$.

When interpreted over the Boolean semiring, the $\Name$ program  in~\eqref{eqn:linear:tc}
is the transitive closure program from Example~\ref{ex:TC}.

When interpreted over the {\em tropical semiring} $\trop^+ = (\R_+ \cup
\{\infty\}, \min, +, \infty, 0)$, the $\Name$ program in~\eqref{eqn:linear:tc}
 solves the classic {\em All-Pairs-Shortest-Path} (APSP) problem, which  computes the shortest path length $T(X,Y)$ between all pairs $X,Y$ of
 vertices in a directed graph specified by an edge relation $E(X,Y)$, where the semiring
 element associated with $E(X,Y)$ is the length of the directed edge $(X,Y)$.
  \begin{align}
    T(X,Y) &\cd \min\left(E(X,Y), \min_Z (T(X,Z) + E(Z,Y))\right) \label{eqn:apsp}
  \end{align}
\end{ex}

A $\Name$ program can be thought of as an immediate consequence
operator (ICO). {A simple way to understand the semantics of $\Name$ is
to think of each body predicate $R_i$ in~\eqref{eq:t:monomial} as a
function from the domain of $\bm X_i$ to the domain $S$ of the semiring.
The functional value $R_i(\bm c_i)$ for a particular binding
$\bm c_i$ in the domain of $ \bm X_i$ is the value assigned to
the ground atom $R_i(\bm c_i)$.
The rule~\eqref{eq:t:monomial} is thus exactly a {\em sum-product query}
 \ (or a {\em functional  aggregate query}~\cite{DBLP:conf/pods/KhamisNR16} over one semiring), and multiple rules  with the same head are combined
into a {\em sum-sum-product} query.
The $\Name$ program containing these queries compute new (IDB)
functions from old (IDB and EDB) functions, using the sums and
products from the semiring.

The iterative evaluation of a $\Name$ program works by  initially assigning all IDB ``functions'' to be identically $0$
(i.e. their ground atoms are assigned with $0$).
The ICO is then repeatedly applied to the current IDB state.
In the context of the $\Name$ program in (\ref{eqn:apsp}),
initially all $T(x,z)$ are assigned with $+\infty$ (the $0$ of the tropical
semiring), and the rule~\eqref{eqn:apsp}  effectively is the well-known
Bellman-Ford algorithm~\cite{DBLP:books/daglib/0023376}.
}

A $\Name$ program is  linear if every rule~\eqref{eq:t:monomial} has
no more than one IDB predicate in its body.
The $\Name$ program in Example \ref{ex:APSP1} is linear. While many natural $\Name$ programs are linear,
there are also natural nonlinear $\Name$ programs.

\begin{ex} \label{ex:APSP3}
As a classic example of a nonlinear $\Name$ program, consider the following
alternate formulation of APSP, which is equivalent to~\eqref{eqn:apsp}
   \begin{align}
    T(X,Y) &\cd \min\left(E(X,Y), \min_Z (T(X,Z) + T(Z,Y))\right) \label{eqn:apspnl}
  \end{align}
\end{ex}

The {\em convergence rate} of a $\Name$ program is the stability index of its ICO.

\subsection{Grounding the ICO}
Since the final associated semiring values of the ground IDB atoms are not initially known, it is natural to think of them as (IDB)  variables.
Then the grounded version of the ICO of a $\Name$ program is a map $\bm f : S^n \rightarrow S^n$,
where  $S$ is the semiring domain, and $n$ is the number of ground IDB atoms that ever have a nonzero value at some point in the iterative evaluation of the program.
For instance, in \eqref{eqn:apsp},
there would be one variable  for each pair $(x, y)$ of vertices where there is a directed path from $x$ to $y$ in the graph.
So the grounded version of the ICO of a $\Name$ program has the following form:%
\begin{align}
  X_1 \cd & f_1(X_1, \ldots, X_n) \nonumber\\
          & \ldots \label{eq:grounded:pi} \\
  X_n \cd & f_n(X_1, \ldots, X_n) \nonumber
\end{align}
where the $X_i$'s are the  IDB variables, and $f_i$ is the component of $\bm f$ corresponding to the IDB variable $X_i$.
Note that each component function $ f_i$ is a multivariate polynomial in the IDB variables
of degree at most the maximum
number of factors in any product in the body of some rule~\eqref{eq:t:monomial} in the $\Name$ program.
After
$q$ iterations of the iterative evaluation of a $\Name$ program, the semiring value  associated with the ground atom corresponding to $X_i$
will be:
\begin{align}
{f}_i^{(q)}(\bm 0)
\end{align}

{
\begin{ex}
Consider the binary recursive formulation in~\eqref{eqn:apspnl}, written over a generic semiring.
\begin{align}
    T(X,Y) &\cd E(X,Y) \oplus \bigoplus_Z T(X,Z) \otimes T(Z,Y) \label{eqn:binary:apsp}
\end{align}
Suppose $\adom =\{1, 2, 3, 4\}$, and the input EDB contains ground EDB atoms
$E(1,2), E(2,3), E(3,4)$, with corresponding (constant) semiring values $e_{12}, e_{23}, e_{34}$.
Then there will be 16 equations and 16 variables in the grounded ICO; For each $a,b \in \{1, \ldots, 4\}$ there will
a variable $X_{ab}$, and an equation of
the form:
\begin{align*}
X_{ab} &\cd e_{ab} \oplus \bigoplus_{i\in [4]} X_{ai} \otimes X_{ib}
\end{align*}
But as many  of these variables will always be $0$; they are ``inactive''
and thus they can effectively be ignored from the grounded ICO formulation.
Thus effectively one can think of the grounded ICO as having the following $6$ variables and 6 equations:
\begin{align*}
X_{12} &\cd e_{12}
&& X_{23} \cd e_{23}\\
X_{13} &\cd X_{12} \otimes X_{23}
&& X_{24} \cd X_{23} \otimes X_{34} \\
X_{14} &\cd X_{12} \otimes X_{24} \oplus X_{13}\otimes X_{34}
&& X_{34} \cd e_{34}.
\end{align*}
\end{ex}
}

\subsection{Context Free Languages}
It is convenient to reason about the formal expansion of $f_i^{(q)}(\bm 0)$  using context-free languages (CFL). {See \cite{Kozen} for an introduction to CFLs.}
To explain this, it is probably best to start with a concrete example.

\begin{ex} \label{ex:fundamental} The following map
  $\bm f = (f_1,f_2)$:
    \begin{align}
        \begin{bmatrix}
            A\\ B
        \end{bmatrix}
        \to
        \begin{bmatrix}
            aAB + bB + c \\
            cAB + bA + a
        \end{bmatrix}
    \end{align}
    can be represented by the following context free grammar $G$:
    \begin{align*}
        A \to aAB \ | \ bB \ | \ c &&
        B \to cAB \ | \ bA \ | \ a
    \end{align*}
    \end{ex}

More generally the variables in $f$ become non-terminals in $G$,
the constants in $f$ become the terminals in $G$,
multiplication in $f$ becomes concatenation in $G$,
and addition in $f$ becomes the or operator $\mid$  in $G$.
Given a parse tree $T$ for the grammar,
define the {\em yield} $Y(T)$ of $T$ to be the string of
terminal symbols at the leaves of $T$, and the  {\em product yield} $Z(T)$ to be the product
of the semiring values in $Y(T)$.
Let $\calT^i_q$ denote the set of all parse
trees  with starting non-terminal $X_i$, and depth $\leq q$.
Note that then:
    \begin{align}
        f_i^{(q)}(0) &= {\bigoplus}_{T \in \calT^i_q} Z(T). \label{eqn:YT}
    \end{align}
That is, the value of the semiring value associated with a IDB variable
$X_i$ after $q$ iterations is
the sum of the product of the leaves of parse trees of depth at most $q$ and rooted at $X_i$.

\subsection{Parikh's Theorem}
The upper bound on the time to convergence for iterative evaluation of $\Name$ programs
in ~\cite{Khamis0PSW22} essentially relied  on black-box application of Parikh's theorem~\cite{MR209093}.  {See \cite{friendlyParikh} for an introduction to Parikh's theorem.}
The {\em Parikh image} of a word $w \in \Sigma^*$, denoted by $\Psi(w)$, is the vector
$\Psi(w) = (k_1,\dots,k_\sigma) \in \N^\sigma$ where $k_i$ is the number of occurrences of the letter $a_i \in \Sigma$ that occur in the word
$w$ (So $|\Sigma| = \sigma$). {Similarly, for a language $L$, define $\Psi(L) := \{ \Psi(w) \; | \; w \in L \}$.}
Then using our assumption that the underlying semiring $\bm S$ is commutative,
we observe that
\begin{align}
    f_i^{(q)}(\bm 0) &= {\bigoplus}_{T \in \calT^i_q} ~ Z(T)
            = {\bigoplus}_{T \in \calT^i_q} ~ {\bigotimes}_{j=1}^\sigma ~  a_j^{\Psi_j(Y(T))}
\label{eqn:YT2}
\end{align}
where $\Psi_j(w)$ is component $j$ of the Parikh image.
One version of Parikh's theorem~\cite{MR209093} states that the Parikh images of the words
in a context free language  forms a semi-linear set.
A set  is then {\em semi-linear} if it is a finite
union of linear sets.
A set $\calL \subseteq \N^\sigma$ is said to be {\em linear} if there exist offset vector $\bm v_0$ and basis vectors $ \bm v_1, \ldots, \bm v_\ell \in \N^\sigma$ such that $\calL$ is the span  of these vectors, that is if:
\begin{align*}
  \calL = & \setof{\bm v_0 + k_1 \bm v_1 + \ldots + k_\ell \bm v_\ell}{k_1, \ldots, k_\ell \in \N}.
\end{align*}
{Here we assume that $0 \in \N$.}
If a vector $$\bm v = \bm v_0 + k_1 \bm v_1 + \ldots + k_\ell \bm v_\ell$$
then we say $(k_1, \ldots, k_\ell)$ is a linear representation of $v$ within $\mathcal L$.

The textbook  proof of Parikh's theorem (see \cite{Kozen}) uses what we will call a wedge.
A wedge within a parse tree $T$ can be specified by identifying  two internal nodes in the parse tree that correspond to the same nonterminal, say $A$, and that have an ancestor-descendent relation. The  corresponding wedge $W$ then consists of
the nodes in the parse tree that are descendents of the top $A$, but not the bottom $A$.
See Figure \ref{fig:wedge}.

\begin{figure}[th]

\centering
\includegraphics[width=0.27\textwidth]{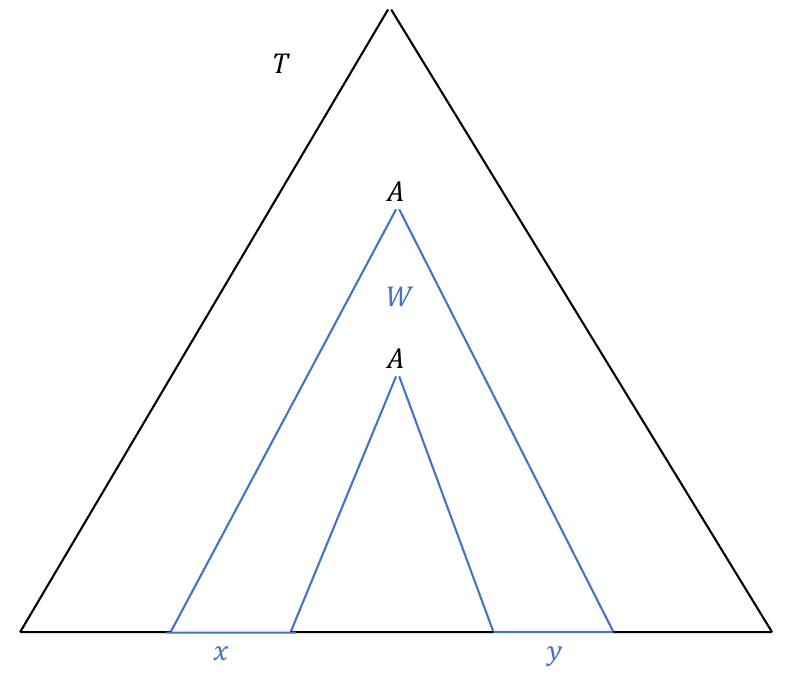}
\caption{Illustration of a wedge $W$}
\label{fig:wedge}
\end{figure}

\section{Technical Overview}
\label{sec:overview}

We now give a technical overview of the
 proof of Theorem \ref{thm:main}. This proof has three logical parts. The first  part, which is stated in
 Theorem \ref{thm:depth-parikh},
 strengthens some aspects of the standard statement of Parikh's theorem~\cite{Kozen}.

\begin{thm}
    \label{thm:depth-parikh}

  Let  $L$ be a context free language generated by a grammar $G= (N, \Sigma, R, S)$, where $N$ is the collection of nonterminals, $\Sigma$ is
    the collection of terminals, $R$ is the collection of rules, and $S \in N$ is the
    start non-terminal. Let $n$ be the cardinality of $N$ and let
    $\lambda$ be the maximum number of symbols on the righthand side of any rule.
 Then
        there exists a finite semi-linear set $\mathcal{M}$ with the following properties:
    \begin{enumerate}
        {\item $\mathcal M = \Psi(L)$.}

        \item Every linear set
        $\mathcal{L} \in \mathcal{M}$ has an associated offset vector $v_0$ and basis vectors $v_1, \ldots v_{\ell}$,
       with the properties that :
        \begin{enumerate}
            \item There is a word $w \in L$ such   $\Psi(w) = v_0$ and $w$ can be generated by a parse tree  with depth at most  $ n(n+3) / 2$.
            \item Each basis vector has an associated wedge with depth at most  $ n(n+3) / 2$.
            \item The 1-norm of the offset vector, and each basis vector, is at most $\lambda^{n(n+3) / 2}$.

            \item For each vector $\bm v = \bm v_0 + k_1 \bm v_1 + \ldots + k_\ell \bm v_\ell$ in the span of $\mathcal{L} \in \mathcal{M}$,
            there is a word $w \in L$, with $\Psi(w) = \bm v$, where $w$ can be generated by a parse tree  with depth at most $ (k+1) n(n+3) / 2$, where $k = k_1 + k_2 + \ldots + k_\ell$.
   \end{enumerate}
    {
    \item For any parse tree $T$ of depth $d$ such that $Y(T)  \in L$, there exists a linear representation of $\Psi(Y(T)) =  \bm v_0 + k_1 \bm v_1 + \ldots + k_\ell \bm v_\ell$ such that $1 + k_1 + k_2 + \ldots + k_\ell \geq d / (n(n+3)/2)$.
}

    \end{enumerate}

\end{thm}

The proof of Theorem \ref{thm:depth-parikh} is
given in Section \ref{sect:Parikh}.
The most important way that Theorem   \ref{thm:depth-parikh} extends the
standard version of Parikh's theorem is property 2(d), which upper bounds
the depth of some parse tree of a word
by the 1-norm of the representation of that word.   The bound given in the textbook proof~\cite{Kozen} of
Parikh's theorem gives a depth bound that is exponentially large.
To achieve property 2(d), our proof contains a constructive forward process $\mathcal P$ that
creates a linear set $\mathcal L \in \mathcal M$
from the parse tree $T$ of some word $w \in L$ by removing wedges from $T$. We are careful to design $\mathcal P$
so that it  is reversible; that is, to recover a parse tree
from a linear representation we can just reverse the process $\mathcal P$.
To accomplish this we need that in the forward process every wedge that is removed does not remove any
nonterminal from the parse tree. Our process $\mathcal P$ ensures that the parse tree for
the offset and the wedges for the basis vectors have depth $O(n^2)$.

The second part of the proof of Theorem \ref{thm:main} is stated in Lemma \ref{lemma:smallrep},
which states that for every word $w \in L$ there exists a linear set $\mathcal L$ in $\mathcal M$ such that
$\Psi(w)  \in \mathcal L$ and $\Psi(w)$ has a linear representation with respect to the basis vectors of $\mathcal L$ that has   small support.

\begin{lmm}
\label{lemma:smallrep}
Let $L$ be an arbitrary context free language. Let $h= 2(\sigma (n(n+3) / 2) \lg (\lambda  +1) + 4\sigma \lg \sigma)$.
Let  $\mathcal M$ be the semilinear set that is guaranteed to exist in Theorem \ref{thm:depth-parikh}.
Let  $w$ be a word in $  L$. Let   $\mathcal L$ in $\mathcal M$ be a linear set such that
$\Psi(w) \in \mathcal L$.  Let
$\Psi(w) = \bm v_0 + k_1 \bm v_1 + \ldots + k_m \bm v_m$ be a linear representation of $\Psi(w)$ with respect to the offset and basis vectors of $\mathcal L$.
Then there exists another linear representation
$\Psi(w) = \bm v_0 + k'_1 \bm v'_1 + \ldots + k'_h \bm v'_h$ of $\Psi(w)$ with respect to $h$ basis vectors of $\mathcal L$ such that $\sum_{i=1}^m k_i = \sum_{i=1}^h k'_i $.
\end{lmm}

The proof of  Lemma \ref{lemma:smallrep}, given in Section \ref{sect:smallsupport}, uses properties of $\mathcal L$ established in our strengthened
version of Parikh's theorem (Theorem \ref{thm:depth-parikh}), and
the pigeon hole principle
to establish that any word in  $\mathcal L$ that has a linear representation with respect to $\mathcal L$
with large support, also has a linear representation with smaller support.
In particular we use the finding that all the offset and basis vectors have (relatively) small 1-norms.
This makes formal the intuition that one might draw from standard vector spaces that
the number of basis vectors needed to represent a vector/word in the span of
some basis vectors
shouldn't be more than the dimensionality of the spanned space.

Finally in Section \ref{sect:mainproof} we use Lemma \ref{lemma:smallrep} to prove Theorem \ref{thm:main}.
{
To show that for sufficiently large $q$ it is the case that $f_i^{(q)}(\bm 0)=f_i^{(q+1)}(\bm 0)$,
we show that for every tree $T' \in \mathcal T^i_{q+1} \setminus \mathcal T^i_q$,
the corresponding summand $Z(T')$  is ``absorbed'' by the earlier terms, that is:
\begin{align}
\bigoplus_{T \in \mathcal T^i_q} Z(T) \oplus Z(T') = \bigoplus_{T \in \mathcal T^i_q} Z(T) \label{eq:absorb}
\end{align}
From property (3) of Theorem~\ref{thm:depth-parikh} we know that that there exists a linear representation
of $\Psi(Y(T'))$ that has a large 1-norm;
and from Lemma~\ref{lemma:smallrep} we know  that there is a linear representation with the same large 1-norm
of $\Psi(Y(T'))$ with small support.
Thus we can conclude by the pigeonhole principle that
one of the basis vectors, in this
small support linear representation, must have a coefficient greater than the stability $p$ of the ground set.
Finally, Eqn. (\ref{eq:absorb}) follows from the stability of the semiring element that
corresponds to the semiring element that is the ``product'' of that basis vector.
}

\newcommand{\cone}{c}

\section{The Strengthened Parikh's Theorem}
\label{sect:Parikh}

\newcommand{\cB}{\mathcal{B}}
\newcommand{\cW}{\mathcal{W}}
\newcommand{\cT}{\mathcal{T}}
\newcommand{\cTs}{\mathcal{T}^S_{\cone}}

Our goal in this section is to prove Theorem     \ref{thm:depth-parikh}.

We begin our proof by defining a  semi-linear set $\mathcal{M}$ with the desired properties. Let $\cone:= n(n+3)/2$ throughout this section.
Let $\cTs$ denote the set of all parse trees (starting with the non-terminal $S$) of depth at most $\cone$. Recall that $T$'s yield, denoted as
$Y(T)$, is the word obtained by a parse tree $T$.
%
 For a wedge $W$, $Y(W)$ is  analogously defined by ignoring the unique non-terminal leaf node in the wedge $W$.
Let $N(T)$ denote the set of non-terminals that appear in $T$.
Let $\cW^A_{\cone}(T)$ be the collection of wedges that appear in $T$, have a non-terminal $A$ as the root,
and have depth (or equivalently height) at most $\cone$. Let $\cB_{\cone}(A, T) := \{ Y(W) \; | \; W \in \cW^{A}_{\cone}(T)\}$.\footnote{{While we use notations $\cW^A_{\cone}(T)$ and $\cB_{\cone}(A, T)$ for notational brevity, their dependence is on $N(T)$ rather than $T$.}} 
For notational brevity, we may use $\cB(A,T)$ instead of $\cB_{\cone}(A,T)$.

Then, for each tree $T$ in $\cTs$, we define a linear set where the offset vector is $\Psi(Y(T))$
and the basis vectors are $\cup_{V \in N(T)} \Psi(\cB(V,T))$.
Here, $\Psi(L')$ denotes the collection of vectors corresponding to the subset of words, $L'$. Notice that because of the way we created
the offset vector and basis vectors, there is a  parse tree in $\cTs$ corresponding to the offset vector and a wedge corresponding to each basis vector, all of depth at most $\cone$.

In the following we recall the definition of wedges (Figure~\ref{fig:wedge}) and define how to index them. For an arbitrary parse tree $T$ we will map it to a tree in $\cTs$ by iteratively removing a wedge.

\begin{defn}
    Define a \emph{wedge} of a parse tree $T$ as follows.  Consider two occurrences of a non-terminal $A$ in $T$ where one is an ancestor of the other. Let $A'$ be the ancestor node and $A''$ the descendant node. The wedge induced by the pair $(A', A'')$ is defined as the subtree rooted at $A'$ with the subtree rooted at $A''$ removed. The wedge is denoted as $W(A', A'')$. The wedge's depth is defined as the maximum number of edges from $A'$ to a leaf node in $W(A', A'')$.
\end{defn}

We would like to keep the following invariant, throughout the iterative process.
\begin{figure}[th]

\centering
\includegraphics[width=0.39\textwidth]{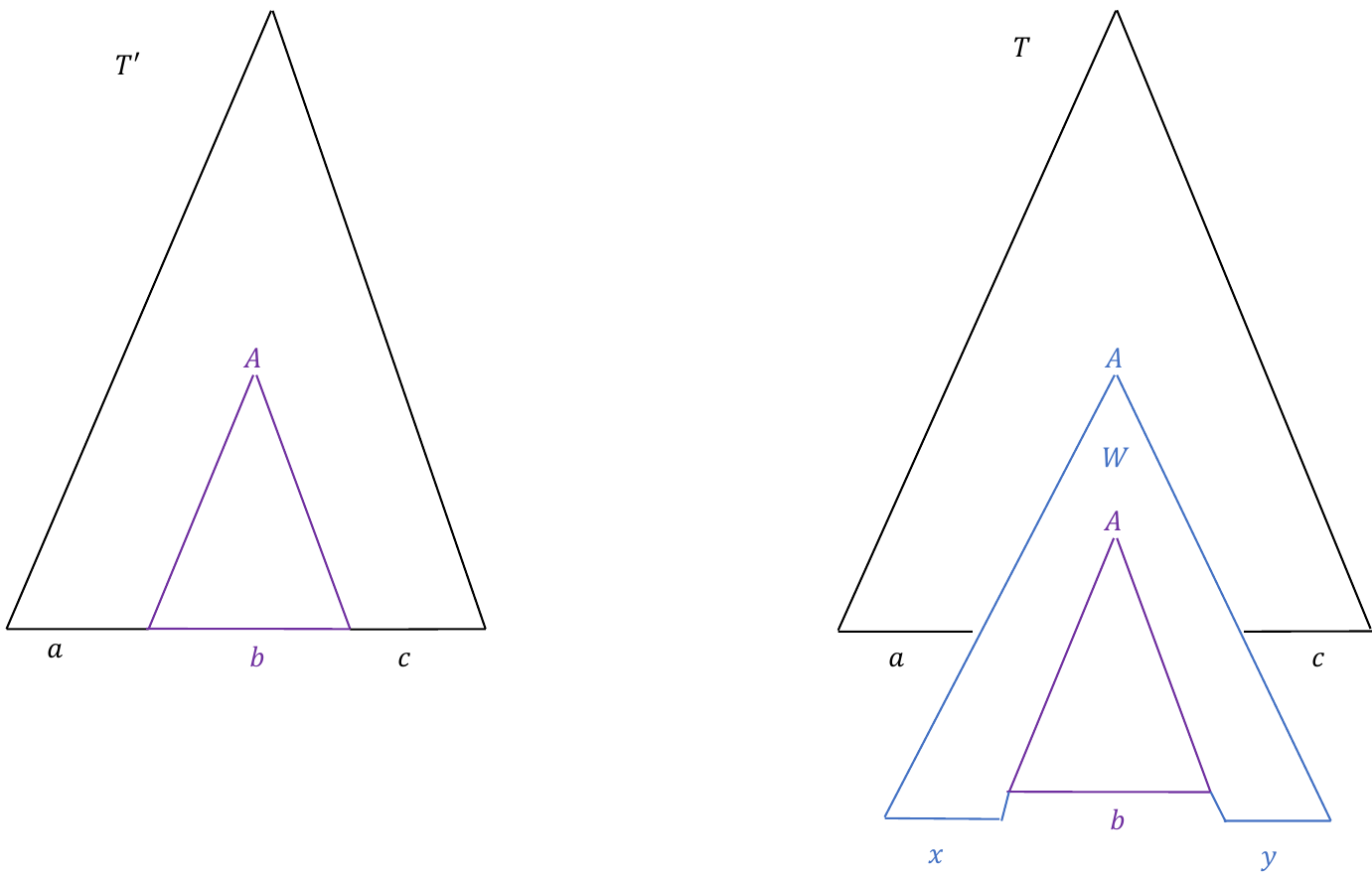}
\caption{The right tree $T$ is recovered from the left tree $T'$ by augmenting the  wedge $W$.}
\label{fig:wedge2}
\end{figure}

\begin{lmm}
    \label{lem:1step-collapse}
    Given a parse tree $T$ of depth greater than $\cone$ starting with non-terminal $S$, we can obtain a parse tree $T'$ starting with $S$ that satisfies the following:
    \begin{enumerate}
        \item (Preserving Non-terminals) $N(T) = N(T')$.
        \item (Reversibility) $T$ can be obtained by replacing one non-terminal  $A$ in $T'$ with a wedge $W \in \cW_{\cone}^{A}(T')$ for some $A \in N(T')$.
    \end{enumerate}
\end{lmm}

Alternatively, the second property means that $T$ can be obtained from $T'$ by augmenting $T'$ with a wedge $W$ of depth at most $\cone$ corresponding to a vector in $\Psi(\cB(A,T'))$ for some non-terminal $A$ in $N(T')$. Here it is worth noting that $\cW^A_{\cone}(T) = \cW^A_{\cone}(T')$ because $N(T) = N(T')$. See Figure~\ref{fig:wedge2} for an illustration of reversibility.

The first property in the lemma is worth special attention. Suppose we obtained $T'$ from $T$ by repeatedly applying the lemma, but without guaranteeing the first property. Suppose
$\{W_1, W_2, \dots \}$ are the wedges we removed in the process; so $W_i$ is the wedge removed in iteration $i$. Then, we may not be able to augment $T'$ with an arbitrary subset of the wedges, which is critical to establish the {$\Psi(L) \supseteq \calM$} direction of the first property of Theorem~\ref{thm:depth-parikh}.

To show Lemma~\ref{lem:1step-collapse}, consider an arbitrary parse tree $T$ of depth more than $\cone$. We show how to obtain $T'$ by collapsing a wedge induced by two occurrences of the same non-terminal. Below, we describe how we find a ``good'' pair of two occurrences of the same non-terminal we want to collapse. We first define what makes pairs good in the following.

\begin{defn}
    For a given parse tree $T$, we say a pair of ascendant and descendant nodes $(A', A'')$ of  the same non-terminal $A$ is good if it satisfies the following:
    \begin{itemize}
        \item Let $T'$ be the tree $T$ with the wedge $W(A', A'')$ removed. We have $N(T) = N(T')$.
        \item The height of $A'$ is at most $\cone$. In other words, the subtree rooted at $A'$ has depth at most $\cone$.
    \end{itemize}

\end{defn}

In the following we will show that a tree of large depth must have  a good pair. Note that we will immediately have Lemma~\ref{lem:1step-collapse} as corollary if we prove the following the lemma.

\begin{lmm}
    A parse tree $T$ of depth at least $\cone$ has a good pair of nodes. 
\end{lmm}


\begin{proof}
We prove the lemma by an induction on the number of non-terminals.
Consider an arbitrary node $v$ of the largest depth, which must be at least $\cone = (n+3)n / 2 \geq (n+1) + n + \ldots + 2$. Since the base case $n = 1$ is trivial, suppose $n\geq 2$. Consider the unique path from $v$ to the root non-terminal $S$ in $T$. 
The height of a node $u$ on the path is defined as the number of nodes below $u$ on the path, including $u$.

{Consider walking from $u$ towards the root. On this path, consider the first time two occurrences of the \emph{same} non-terminal $V_1$ appear. Say they appear at nodes $v(V_1)$ and $u(V_1)$, where $u(V_1)$ is an ascendant of $v(V_1)$. Observe that $u(V_1)$ has height at most $n+1$ due to the pigeon hole principle. This is because some non-terminal must repeat among $n+1$ nodes. 
}



{If $(u(V_1), v(V_1))$ is a good pair, we are done. If not, it means that the wedge $W(u(V_1), v(V_1))$ must include a non-terminal that \emph{doesn't appear anywhere else in the tree.}   Consider the tree $T_2$ with the subtree rooted at $u(V_1)$ removed. This subtree $T_2$ has at most $n-1$ non-terminals and has  depth at least $n + (n-1) + \ldots + 2$. By induction, this implies that $T_2$ must have a good pair $(u_2, v_2)$. }

{Finally, it is easy to see that the pair remains to be good with respect to $T$ as well: 
 First, $u_2$ has height at most $(n+1) + (n + (n-1) + \ldots + 2) = c$ in the tree $T$. Second, the wedge indexed by $(u_2, v_2)$ does not intersect the subtree rooted at $u(V_1)$, and therefore, the set of non-terminals remains unchanged after removing the wedge from $T$, just as it does when we remove the wedge from $T_2$.
}
\end{proof}


We are now ready to prove Theorem~\ref{thm:depth-parikh}.

\noindent
\textbf{Property (1) { $\mathcal M \supseteq \Psi(L)$ and Property (3)}.}
Given a parse tree $T$ for $w \in L$, suppose we obtained a sequence of trees $T_0 = T, T_1, \ldots, T_\eta$ by repeatedly applying
Lemma~\ref{lem:1step-collapse}, where $T_\eta \in \cT^S_{\cone}$ and $T_i$ is obtained from $T_{i-1}$ by deleting
a wedge $W_i$ in $\cW^{A}_{\cone}(T)$ for some non-terminal $A$ in $N(T) = N(T_\eta)$; note $\cW^{A}_{\cone}(T) = \cW^{A}_{\cone}(T_\eta)$ since $N(T) = N(T_1) = \ldots = N(T_{\eta})$.
 Let $\bm b_i = \Psi(Y(W_i))$. Clearly,
$\Psi(w)$ can be expressed as $\Psi(Y(T_\eta)) + \sum_{i = 1}^{\eta} \bm b_i$.

Since we created a linear set for each parse tree in $\cTs$, thus for $T_\eta$, this is a linear representation within $\calL$ which consists of offset vector $\Psi(Y(T_\eta))$ and basis vectors $\cup_{V \in N(T_\eta)} \Psi(\cB(V,T_\eta))$. This proves { $\mathcal M \supseteq \Psi(L)$} part of the first property.

{
Property (3) is immediate from the above: Suppose that  the parse tree $T$ has depth $d$. The two trees $T_{i-1}$ and $T_i$ have depths differing by at most $c$ since we obtained $T_{i}$ from $T_{i-1}$ by deleting a wedge of depth at most $c$. Further, the last tree $T_{\eta}$ has depth at most $c$ as well. Thus, $c (\eta +1) \geq d$. Since $\eta = k$ where $k$ is described as in the theorem, we have proven this property.
}


\noindent
\textbf{Property (1) { $\mathcal M \subseteq \Psi(L)$} and Property (2)(d).}
Conversely, suppose $w$ has a linear representation within some $\calL \in \mathcal{M}$.
Say the linear representation is $\Psi(Y(T')) + \sum_{i = 1}^k \bm b_i$ for some $T' \in \cT_c^{S}$.
Note that for each basis vector $\bm b_i$, there exists $V_i \in N(T')$ such that $\bm b_i \in \Psi(\cB(V_i,T'))$.
Because of the way we defined linear sets, $\bm b_i = \Psi(Y(W_i))$ for some $W_i \in \cW_{\cone}^{V_i}(T')$. We can augment $T'$ with the wedges $W_i$  in an arbitrary order.
Adding each wedge $W_i$ increases the tree depth by at most $\cone$. {By repeating this for each $\bm b_i$, we obtain a parse tree $T$ such that $\Psi(Y(T)) = \Psi(w)$. This proves $\mathcal M \subseteq \Psi(L)$ part of Property (1). Furthermore, Property (2)(d) follows since $T$ has depth at most at most $\cone (k+1)$. }
%

\noindent
\textbf{Property (2)(a,b,c).} By definition of the offset and basis vectors, it immediately follows that their depth is at most $\cone$. Furthermore, the 1-norm of any of them is at most $\lambda^{\cone}$ because each node has at most $\lambda$ children.

%
\medskip

Together, the proof of the above properties  prove Theorem     \ref{thm:depth-parikh}.

\section{Small Support Representations}
\label{sect:smallsupport}

{
In this section we prove Lemma \ref{lemma:smallrep}. The lemma shows that any linear representation of the
Parikh image of a  word in a
context-free language $L$ can be converted into another linear representation with equal 1-norm, but  with small support. 

\begin{lmm}{[Lemma~\ref{lemma:smallrep} Restated]}
Let $L$ be an arbitrary context free language. Let $h := 2(\sigma (n(n+3) / 2) \lg (\lambda  +1) + 4\sigma \lg \sigma)$.
Let  $\mathcal M$ be the semilinear set that is guaranteed to exist in Theorem \ref{thm:depth-parikh}.
Let $w$ be a word in $  L$.
Let $\mathcal L$ in $\mathcal M$ be a linear set such that
$\Psi(w) \in \mathcal L$.  Let
$\Psi(w) = \bm v_0 + k_1 \bm v_1 + \ldots + k_m \bm v_m$ be a linear representation of $\Psi(w)$ with respect to the offset and basis vectors of $\mathcal L$.
Then there exists another linear representation
$\Psi(w) = \bm v_0 + k'_1 \bm v'_1 + \ldots + k'_h \bm v'_h$ of $\Psi(w)$ with respect to $h$ basis vectors of $\mathcal L$ such that $\sum_{i=1}^m k_i = \sum_{i=1}^h k'_i $.
\end{lmm}}

\begin{proof}
To streamline our analysis, we will assume that $\lambda \geq 2$. This is without loss of generality because in  the case that $\lambda =1$ we can add an unused terminal to $\Sigma$. The value of $\lambda$ will increase by one in the final bound for this boundary case.
{
For an arbitrary word $w$ in our language $L$, suppose we are given a linear representation of $\Psi(w)$,
$$\bm v_0 + k_1 \bm v_1 + \ldots + k_m \bm v_m,$$
 where
 $$k_1, k_2, \ldots k_m > 0 \textnormal{ and }$$
 $$m > h:= 2(\sigma (n(n+3) / 2) \lg \lambda + 4\sigma \lg \sigma).$$
 It suffices to find another linear representation of $\Psi(w)$,
 $$\bm v_0 + k'_1 \bm v_1 + \ldots + k'_m \bm v_m$$ such that
 \begin{equation}
    \label{eqn:wewant1}
 k_1 + k_2 + \ldots + k_m = k'_1 + k'_2 + \ldots + k'_m \textnormal{ and }
 \end{equation}
 \begin{equation}
    \label{eqn:wewant1}
  k'_i = 0 \textnormal{ for some } i \in [m]:= \{1, 2, \ldots, m\}.
 \end{equation}
}

A key step to our proof is showing that there exist two distinct
subsets $H_1 $ and $H_2$ of {$[m] :=\{1, 2, \ldots, m\}$} such that
\begin{equation}
    \label{eqn:H12}
    \sum_{i \in H_1} \bm v_i =  \sum_{i \in H_2}  \bm v_i
\end{equation}
We will prove this claim by proving that
\begin{align}
K &:= |\{ \sum_{i \in H} \bm v_i \; | \; \emptyset \neq H \subseteq [m]    \}| < 2^m.
\label{eqn:K:bound}
\end{align}
Note that $K$ is the number of distinct vectors we can generate by summing
a non-empty subset of vectors from $\bm v_1, \bm v_2, \ldots, \bm v_m$.
The existence of the desired pair of $H_1$ and $H_2$ satisfying~\eqref{eqn:H12} will then
follow from pigeonhole principle.

Let $M$ be the maximum 1-norm of any $\bm v_i$, i.e., $||\bm v_i||_1 \leq M$, for all $i \in [m]$. Thus, we have $||\sum_{i \in H}  \bm v_i||_1 \leq Mm$ for any $H \subseteq [m]$. We use the following well-known fact: the number of distinct vectors in
$\mathbb{N}^d$ with 1-norm of $k$ is exactly $\binom{k +d -1}{d-1} \leq (k+d-1)^{d-1}$ (see~\cite{MR2868112}). In our case, $k \leq Mm$ and $d = \sigma$. Thus, we have $K \leq (Mm + \sigma -1)^{\sigma-1} (Mm +1) \leq (Mm+ \sigma)^\sigma$.
If $\sigma = 1$ (there exists only one terminal), we have a tighter bound of $K \leq M + (M-1) + \ldots + M - (m-1) = m(2M - (m-1))/2$.

{
To prove~\eqref{eqn:K:bound}, it remains to show that $(Mm+\sigma)^\sigma < 2^m$ when $\sigma \geq 2$ and that $h(2M-m+1)/2 < 2^m$ when $\sigma=1$.
We consider two cases: $\sigma \geq 2$ and $\sigma = 1$.}

\smallskip
\noindent{
\textbf{Case i: $\sigma \geq 2$.}}
We shall now establish
\begin{equation}
    \label{eqn:pigeon}
 2^{m}>  (Mm+\sigma)^\sigma \textnormal{ when $\sigma \geq 2$}
\end{equation}

By taking the logarithm, the inequality (\ref{eqn:pigeon}) is equivalent to:
\begin{align}
    m > \sigma \lg (M m+\sigma) \label{eqn:pigeon2}
\end{align}
{
We know from Theorem \ref{thm:depth-parikh} (c) that $M \le \lambda^{n(n+3)/2}$.
Thus we know that the following equation would imply Eqn.
(\ref{eqn:pigeon2}):
\begin{align}
  m > \sigma \lg (m \lambda^{n(n+3)/2} +\sigma ) \label{eqn:pigeon3}
\end{align}

Using the fact that $\lg x$ is sub-additive when $x \geq 2$ and the assumptions that $\lambda, \sigma \geq 2$, we have
\begin{align}
   \sigma \lg (m \lambda^{n(n+3)/2}) + \sigma \lg \sigma > \sigma \lg (m \lambda^{n(n+3)/2} +\sigma ) \label{eqn:pigeon4}
\end{align}
}

Thus, it is sufficient to show:
\begin{align}
  & \;\; m > \sigma \lg (m \lambda^{n(n+3)/2}) + \sigma \lg \sigma \nonumber \\ & \;\;= \sigma \lg m + \sigma (n(n+3) / 2) \lg \lambda + \sigma \lg \sigma \nonumber \\
   \Leftrightarrow & \;\; m - \sigma \lg m>   \sigma (n(n+3) / 2) \lg \lambda + \sigma \lg \sigma  \label{eqn:pigeon5}
\end{align}

We show that
\begin{equation*}
m - \sigma \lg m \geq m / 2
\end{equation*}
when $m \geq 8 \sigma \lg \sigma$: Since $m/2 - \sigma \lg m$ is increasing in $m$ when $m \geq 8 \sigma \lg \sigma$, we have
$m/2 - \sigma \lg m \geq 4 \sigma \lg \sigma - \sigma \lg (8 \sigma \lg \sigma)
= \sigma \lg (\sigma^3 / (8\lg \sigma)) \geq 0$ when $\sigma \geq 2$, as desired. 

Thus, we have
\begin{equation}
    \label{eqn:pigeon6}
m - \sigma \lg m \geq m / 2 > \sigma (n(n+3) / 2) \lg \lambda + 4 \sigma \lg \sigma,
\end{equation}
{
where the second inequality follows from the fact that $m > h$. From Eqn.
(\ref{eqn:pigeon}),
(\ref{eqn:pigeon2}),
(\ref{eqn:pigeon3}), (\ref{eqn:pigeon4}), (\ref{eqn:pigeon5}), and (\ref{eqn:pigeon6})
we have $2^m > K$ when $\sigma \geq 2$.}

{
\smallskip
\noindent
\textbf{Case ii: $\sigma = 1$.}}
If $\sigma = 1$, as mentioned above,
we have
\begin{align*}
K & \leq M + (M-1) + \ldots + M - (m-1) \\
&= m(2M - (h-1))/2  \\
&< Mm \\
&\leq M2^{m/2} \quad \mbox{[\text{Since $m\geq 2$}]} \\
& \leq \lambda^{n(n+3)/2} 2^{m/2}  \\
&\leq 2^m.
\end{align*}
The last inequality is true due to the assumption that  $\sigma = 1$ and
\begin{align*}
m &> h = 2(\sigma (n(n+3) / 2) \lg \lambda + 4\sigma \lg \sigma) \geq n (n+3) \lg \lambda
\end{align*}

\smallskip
{
Thus, we have shown that $2^m > K$ for all $\sigma \geq 1$, which establishes the existence of $H_1 \neq H_2 \subseteq [m]$ satisfying Eqn. (\ref{eqn:H12}).}

We now explain how to
construct a new representation of $\Psi(w)$ that contains less basis vectors.
First observe that one of two sets $H_1, H_2$ doesn't contain the other since no basis vectors are $\bm 0$ and we have Eqn.~(\ref{eqn:H12}).
Let $i$ be $\arg \min_{i' \in H_1 - H_2} k_{i'}$, breaking ties arbitrarily.
Then for $j \in H_1-H_2$ let $k_j' = k_j - k_i$,
for $j \in H_2- H_1$ let $k_j' = k_j + k_i$,
and for all other $i$ let $k_j' = k_j$.

Note that there was no change in the sum, i.e.,
\begin{align}
 w = & \; \bm v_0 + k_1 \bm v_1 + \ldots + k_m \bm v_m \nonumber \\
  =  & \; \bm v_0 + k'_1 \bm v_1 + \ldots + k'_m \bm v_m \label{eqn:1-step}
\end{align}
However, this new representation $\langle k'_1, k'_2, \ldots, k'_m \rangle$ has a strictly  smaller support since $k'_i = 0$.
{
Observe that $k_1 + k_2 + \ldots + k_m = k'_1 + k'_2 + \ldots + k'_m$. Thus, we have found another linear representation $\bm v_0 + k'_1 \bm v_1 + \ldots + k'_m \bm v_m$ of $\Psi(w)$ that has a smaller support than the given linear representation
$\bm v_0 + k_1 \bm v_1 + \ldots + k_m \bm v_m$ preserving the 1-norm value in the linear representation.

We can repeat this process until we obtain a linear representation of support size at most $h$.} Finally, recall that we assumed $\lambda \geq 2$. To remove this assumption, as mentioned at the beginning, we can add an unused terminal to $\Sigma$, which increments the value of $\lambda$ by one in the bound.
\end{proof}

\section{Bounding the Number of Iterations}
\label{sect:mainproof}

This section is devoted to proving Theorem \ref{thm:main}, restated here.

{
\begin{thm} [Theorem~\ref{thm:main} Restated ]
Let $\bm S$ be a $p$-stable commutative semiring.
Let $P$ be a $\Name$ program where the maximum number of multiplicands
in any product is at most $\lambda$. Let $D$ be the input EDB database.
Let $\sigma$  be number of the semiring elements referenced in $P$ or $D$.
Let $n$ denote the total number of ground atoms in an IDB
that at some point in the iterative evaluation of $P$ over semiring $\bm S$ on input $D$ have
a nonzero associated semiring value.
Then the iterative evaluation of $P$ over semiring $\bm S$ on input $D$
    converges within
    $$\lceil  p n(n+3) \cdot (\sigma (n(n+3) / 2) \lg (\lambda+1) + 4\sigma \lg \sigma + 1) \rceil$$
        steps.
\end{thm}

 For a vector $\bm v \in \mathbb{N}^\sigma$, we let $Z(\bm v)$ denote the product corresponding to $\bm v$, i.e. $\prod_{s = 1}^\sigma a_s^{\bm v_s}$, where $a_s$ is the element corresponding to the $s$th entry of the vector. We naturally extend the notation to a vector set $\bm V$, by letting $Z(\bm V) := \bigoplus_{\bm v \in \bm V} Z(\bm v)$.
To prove the theorem, we need the following lemma, which roughly speaking shows that the summation of all products corresponding to vectors in $\calL$ with coefficients up to $p$ doesn't change when added a product corresponding to any other vector in $\calL$. This lemma was proven in \cite{Khamis0PSW22} (See Section 5.2, in particular the proof of Theorem 5.10 in the
journal / ArXiV version of \cite{Khamis0PSW22}) but we include the proof in the appendix for completeness.

\newcommand{\calG}{\mathcal G}
\begin{lmm}
    \label{lem:small-enough}
    Let $\calL$ be a linear set with offset vector $\bm v_0$ and basis vectors $\bm v_1, \bm v_2, \ldots \bm v_m$. Let $\calL_{\leq p} := \{ \bm v_0 + \kappa_1 \bm v_1 + \kappa_2 \bm v_2 + \ldots + \kappa_m \bm v_m \; | \; \kappa_i \in [0, p] \forall \ i\}$.
    Consider an arbitrary $\bm w
= \bm v_0 + k_1 \bm v_1 + k_2 \bm v_2 + \ldots + k_m \bm v_m$ where  $(k_1, k_2, \ldots, k_m) \in \mathbb{N}^m$ and $k_i > p$ for some $i$. Then, we have  
$$Z(\calL_{\leq p}) = Z(\calL_{\leq p})  \oplus Z(\bm w)$$
\end{lmm}
We now have all tools to prove Theorem~\ref{thm:main}.  Consider an arbitrary IDB variable $X_r$ and let $L$ be the CFL associated with this variable.
Let $\mathcal M$ be a semi-linear set that satisfies  the properties stated in Theorem~\ref{thm:depth-parikh}.
Let $\cT^r_{q}$ denote the collection of the parse trees of depth at most $q$ starting with $X_r$.  Our goal is to show:
\begin{equation}
    \label{eqn:final-goal}
    f_r^{(q)}(\bm 0) = f_r^{(q+1)}(\bm 0)
\end{equation}
where
$$f_r^{(q)}(\bm 0) =  \bigoplus_{T \in \calT^r_q} Z(T), \textnormal{ and}$$
\begin{equation}
    \label{eqn:final-bound}
   q := \lceil p n(n+3) \cdot (\sigma (n(n+3) / 2) \lg (\lambda+1) + 4\sigma \lg \sigma + 1) \rceil
\end{equation}

Consider an arbitrary $T \in \calT^r_{q+1} \setminus \calT^r_{q}$. By Theorem~\ref{thm:depth-parikh} (3), $Y(T)$ has a linear representation
$$\Psi(Y(T)) = \bm v_0+ k_1 \bm v_1 + k_2 \bm v_2 + \ldots + k_m \bm v_m$$
within some $\mathcal L$ in $\mathcal M$ such that $1 + k_1 + k_2 + \ldots + k_m > q / (n(n+3)/2)$. Thus, we have $k := k_1 + k_2 + \ldots + k_m > ph$ where
$$h := 2(\sigma (n(n+3) / 2) \lg (\lambda +1) + 4\sigma \lg \sigma)$$
By Lemma~\ref{lemma:smallrep},  we can find a linear representation
$\Psi(Y(T))= \bm v_0+ k'_1 \bm v'_1 + k'_2 \bm v_2 + \ldots + k'_h \bm v'_h$,
where $k'_1 + k'_2 + \ldots + k'_h > ph$. By the pigeonhole's principle, we have that $k'_j > p$ for some $j$.

Let $\calL'$ be the subset of $\calL$ that only consists of basis vectors $\bm v'_1, \bm v'_2,  \ldots, \bm v'_h$ together with offset vector $\bm v_0$. Then, by Lemma~\ref{lem:small-enough}, we have
$$\bigoplus_{\bm u \in \calL'_{\leq p}}   Z(\bm u) = \bigoplus_{\bm u \in \calL'_{\leq p}}   Z(\bm u) \oplus Z(T),$$
where we used $\bigoplus_{\bm u \in \calL'_{\leq p}} Z(\bm u)  = Z(\calL'_{\leq p})$, which is the case by definition.
To complete the proof of Theorem~\ref{thm:main}, it is sufficient to show
\begin{equation}
    \label{eqn:final-subset}
\calL'_{\leq p} \subseteq \{ \Psi(Y(T)) \; | \; T  \in \calT_q^r\}
\end{equation}
 To see this consider any $\bm v = \bm v_0 + \kappa_1 \bm v_1 + \ldots + \kappa_h \bm v_h \in \calL'_{\leq p}$. By definition of $\calL'_{\leq p}$, $\kappa_i \leq p$ for all $i \in [h]$. Then, thanks to Theorem~\ref{thm:depth-parikh} property (2)(d), we know that there is a word $w \in L$ with $\Psi(w) = \bm v$ such that $w$ is generated by a parse tree $T'$ of depth at most $(ph +1)n(n+3) / 2 \leq q$. Thus, it must be the case that $\bm  v = Z(T') \in \{ Z(T) \; | \; T  \in \calT_q^r\}$. This establishes Eqn. (\ref{eqn:final-subset}) as desired, and therefore we have proven Theorem~\ref{thm:main}.
}

\section{Conclusion}
\label{sec:conclusions}
{
This paper considers the convergence of recursive $\Name$ programs using natural iterative evaluation over
the semirings where convergence is not program dependent, namely the stable commutative semirings.
Previously the best-known bound on convergence time was exponential in  the output size.
Our main contribution is to show that in fact the time to convergence can be bounded by a polynomial
in the natural parameters, such as the output size.  One consequence of this result is a better
understanding of how much worse the time to convergence can be for general $\Name$ programs than
linear $\Name$ programs. One reasonable interpretation of our results is that the worst-case time to convergence
for general $\Name$ programs is not too much worse than the worst-case time to convergence for linear
$\Name$ programs, which was a bit surprising to us given that generally one doesn't
expect algorithmic convergence bounds for  non-linear optimization to be competitive with the
bounds for linear optimization.

There are several natural directions for followup research.
While essentially tight bounds are known for convergence time for linear $\Name$ programs,
we do not establish the tightness of our bound. So one natural research direction is
to determine tight bounds on the convergence rate for general $\Name$ programs.
Another natural research direction is to  show some sort of bounds on  convergence time
over non-stable semirings. Note that such bounds would have to be program-dependent.
Another natural research direction would be to develop other algorithms for evaluating
$\Name$ programs and analyze their convergence bounds.}

\bibliographystyle{siam}
\bibliography{main}

\begin{thebibliography}{10}

\bibitem{DBLP:books/aw/AbiteboulHV95}
{\sc S.~Abiteboul, R.~Hull, and V.~Vianu}, {\em Foundations of Databases},
  Addison-Wesley, 1995.

\bibitem{DBLP:conf/pods/KhamisNR16}
{\sc M.~{Abo Khamis}, H.~Q. Ngo, and A.~Rudra}, {\em {FAQ:} questions asked
  frequently}, in Proceedings of the 35th {ACM} {SIGMOD-SIGACT-SIGAI} Symposium
  on Principles of Database Systems, {PODS} 2016, San Francisco, CA, USA, June
  26 - July 01, 2016, T.~Milo and W.~Tan, eds., {ACM}, 2016, pp.~13--28.

\bibitem{DBLP:books/daglib/0023376}
{\sc T.~H. Cormen, C.~E. Leiserson, R.~L. Rivest, and C.~Stein}, {\em
  Introduction to Algorithms, 3rd Edition}, {MIT} Press, 2009.

\bibitem{DBLP:journals/jacm/EsparzaKL10}
{\sc J.~Esparza, S.~Kiefer, and M.~Luttenberger}, {\em Newtonian program
  analysis}, J. {ACM}, 57 (2010), pp.~33:1--33:47.

\bibitem{semiring_book}
{\sc M.~Gondran and M.~Minoux}, {\em Graphs, dioids and semirings}, vol.~41 of
  Operations Research/Computer Science Interfaces Series, Springer, New York,
  2008.
\newblock New models and algorithms.

\bibitem{im2023convergence}
{\sc S.~Im, B.~Moseley, H.~Ngo, and K.~Pruhs}, {\em On the convergence rate of
  linear datalogo over stable semirings}, 2023.

\bibitem{Khamis0PSW22}
{\sc M.~A. Khamis, H.~Q. Ngo, R.~Pichler, D.~Suciu, and Y.~R. Wang}, {\em
  Convergence of datalog over (pre-) semirings}, in {PODS} '22: International
  Conference on Management of Data, Philadelphia, PA, USA, June 12 - 17, 2022,
  L.~Libkin and P.~Barcel{\'{o}}, eds., {ACM}, 2022, pp.~105--117.

\bibitem{friendlyParikh}
{\sc C.~Koch}, {\em A friendly tour of parikh’s theorem}.

\bibitem{Kozen}
{\sc D.~C. Kozen}, {\em Automata and Computability}, Springer-Verlag, Berlin,
  Heidelberg, 1997.

\bibitem{MR1470001}
{\sc W.~Kuich}, {\em Semirings and formal power series: their relevance to
  formal languages and automata}, in Handbook of formal languages, {V}ol. 1,
  Springer, Berlin, 1997, pp.~609--677.

\bibitem{DBLP:journals/tcs/Lehmann77}
{\sc D.~J. Lehmann}, {\em Algebraic structures for transitive closure}, Theor.
  Comput. Sci., 4 (1977), pp.~59--76.

\bibitem{MR209093}
{\sc R.~J. Parikh}, {\em On context-free languages}, J. Assoc. Comput. Mach.,
  13 (1966), pp.~570--581.

\bibitem{MR1059930}
{\sc G.~Rote}, {\em Path problems in graphs}, in Computational graph theory,
  vol.~7 of Comput. Suppl., Springer, Vienna, 1990, pp.~155--189.

\bibitem{MR2868112}
{\sc R.~P. Stanley}, {\em Enumerative combinatorics. {V}olume 1}, vol.~49 of
  Cambridge Studies in Advanced Mathematics, Cambridge University Press,
  Cambridge, second~ed., 2012.

\end{thebibliography}

\appendix

\section{Omitted Proof}

\begin{proof}[Proof of Lemma~\ref{lem:small-enough}]
Assume wlog that $k_1,  \ldots, k_{m'} > p$ and $k_{m'+1}, \ldots, k_{m} \leq p$.
%
Let $G_j := \{\bm \kappa = (\kappa_1, \ldots, \kappa_m \; | \; \kappa_i \in [0, k_i] \ \forall i \in [0, j] \cup [m'+1, m] \textnormal{ and }
\kappa_i \in [0, p] \ \forall i \in [j+1, m'] \}$ for all $j \in [0, m']$.
Note that to prove the lemma it suffices to show
$$Z(G_0) = Z(G_0) \oplus Z(\bm w)$$
because $\{\bm v_0 + \kappa_1 \bm v_1 + \ldots + \kappa_m \bm v_m \; | \; \bm \kappa \in G_0\} \subseteq \calL_{\leq p}$. We are going to establish
\begin{align}
    \label{eqn:Gs}
Z(G_0) = Z(G_1) = \ldots = Z(G_{m'})
\end{align}
and
\begin{align}
    \label{eqn:Gs2}
Z(G_{j}) \oplus Z(G_{j+1} \setminus G_{j})  = Z(G_{j}) \; \forall j \in [0, m'-1]
\end{align}

Indeed if we have them, $$Z(G_0) \oplus Z(\bm w) = Z(G_{m'-1}) \oplus Z(\bm w) = Z(G_{m'-1}) = Z(G_0),$$ as desired, since $(k_1, k_2, \ldots, k_m) \in G_{m'} \setminus G_{m'-1}$ and
$\bm w = \bm v_0 + k_1 \bm v_1 + \ldots + k_m \bm v_m$.

It now remains to show Eqn. ~(\ref{eqn:Gs}) and (\ref{eqn:Gs2}).
Consider a fixed  $j \in [0, m'-1]$.
Consider an arbitrary $\bm \kappa \in G_{j}$. Let $\bm \kappa'(q)$ be $\bm \kappa$ with $\bm \kappa_{j+1}$ ($j+1$-th coordinate of $\bm \kappa$) replaced with $q$. Let $\bm v(\bm \kappa'(q)) := \bm v_0 + \kappa'_1(q) \bm v_1 + \ldots + \kappa'_m(q) \bm v_m
$ be the vector represented by the linear representation $\bm \kappa'(q)$.
Let $z_i := Z(\bm v_i)$.
 Then, we have
\begin{align*}
\bigoplus_{q = 0}^p Z(\bm v(\bm \kappa'(q))) &=
\bigoplus_{q = 0}^p
\left(\prod_{i=1: i\neq j+1}^{m} z_i^{\kappa_i} \right) z_{j+1}^{q} =
\prod_{i=1: i\neq j+1}^{m} z_i^{\kappa_i} \bigoplus_{q = 0}^p
z_{j+1}^{q}
\\
&=
\prod_{i=1: i\neq j+1}^{m} z_i^{\kappa_i}
z_{j+1}^{(p)}
\\
&= \prod_{i=1: i\neq j+1}^{m} z_i^{\kappa_i}
z_{j+1}^{(k_{j+1})}   \quad \mbox{[$p$-stability and $k_{j+1} > p$]}
\\
%
&=
\bigoplus_{q = 0}^p
\left(\prod_{i=1: i\neq j+1}^{m} z_i^{\kappa_i}\right) z_{j+1}^{q} +
\bigoplus_{q = p+1}^{k_{j+1}}
\left(\prod_{i=1: i\neq j+1}^{m} z_i^{\kappa_i} \right) z_{j+1}^{q} \\
&=
\bigoplus_{q = 0}^p Z(\bm v(\bm \kappa'(q))) + \bigoplus_{q = p+1}^{k_{j+1}} Z(\bm v(\bm \kappa'(q)))
\end{align*}

%
Since $\bigcup_{q = 0}^p \bm v(\bm \kappa'(q)) \subseteq G_j$, and any vector in  $G_{j+1} \setminus G_j$ is of the form of $\bm \kappa'(q)$ for some $q \in [p+1, k_{j+1}]$ for some $\bm \kappa \in G_j$,
we have $Z(G_j) = Z(G_{j+1})$. In other words, we showed that any product corresponding to $G_{j+1} \setminus G_j$ is subsumed by some $p+1$ products in $Z(G_j)$ using the $p$-stability. For the same reason, we have Eqn.~(\ref{eqn:Gs2}).
%
%
\end{proof}

\end{document}